\documentclass[conference]{IEEEtran}
\IEEEoverridecommandlockouts

\usepackage{amsmath,amssymb,amsfonts}
\usepackage{algorithmic}
\usepackage{graphicx}
\usepackage{textcomp}
\usepackage{xcolor}
\usepackage{caption}
\usepackage{subcaption}
\usepackage{graphicx}
\usepackage{etoolbox}
\usepackage{float}
\newtheorem{example}{Example} % Defines the example environment
\newtheorem{definition}{Definition}
\newtheorem{theorem}{Theorem}

\newtheorem{proof}{Proof}
\newtheorem{observation}{Observation}
\usepackage[ruled,vlined]{algorithm2e}
\usepackage{tikz}
\usepackage{dblfloatfix}
\usepackage{url}
\usepackage{xurl}
\usepackage{comment}
\usepackage{hyperref}
\usepackage{enumitem}
\usepackage{makecell}
\usepackage{booktabs}
\usepackage{multirow}
\UseRawInputEncoding

\BeforeBeginEnvironment{appendices}{\clearpage}

\usepackage{cite}
\def\BibTeX{{\rm B\kern-.05em{\sc i\kern-.025em b}\kern-.08em
    T\kern-.1667em\lower.7ex\hbox{E}\kern-.125emX}}
\begin{document}

\title{From Statistical Disclosure Control to Fair AI: Navigating Fundamental Tradeoffs in Differential Privacy}

\author{\IEEEauthorblockN{Adriana Watson}
\IEEEauthorblockA{\textit{School of Engineering Technology} \\
\textit{Purdue University}\\
West Lafayette, United States \\
watso213@purdue.edu}}

\maketitle

\begin{abstract}
Differential privacy has become the gold standard for privacy-preserving machine learning systems. Unfortunately, subsequent work has primarily fixated on the privacy-utility tradeoff, leaving the subject of fairness constraints undervalued and under-researched. This paper provides a systematic treatment connecting three threads: (1) Dalenius's impossibility results for semantic privacy, (2) Dwork's differential privacy as an achievable alternative, and (3) emerging impossibility results from the addition of a fairness requirement. Through concrete examples and technical analysis, the three-way Pareto frontier between privacy, utility, and fairness is demonstrated to showcase the fundamental limits on what can be simultaneously achieved. In this work, these limits are characterized, the impact on minority groups is demonstrated, and practical guidance for navigating these tradeoffs are provided. This forms a unified framework synthesizing scattered results to help practitioners and policymakers make informed decisions when deploying private fair learning systems.
\end{abstract}

\section{Introduction: The Central Problem}

Consider the following scenario: a hospital wants to publish statistics about patient health outcomes to advance medical research, but must protect individual patient privacy. What information can be released and what should remain private?

While this scenario may seem simple, further analysis reveals a layer of complexity as the challenge is not merely technical but conceptual: \emph{what does it mean for a data release to be ``private''?} Early approaches to this problem focused on removing direct identifiers (such as names, social security numbers, etc). While this is an intuitive solution, it is insufficient as auxiliary information can enable re-identification. This highlights a fundamental concern: can we define privacy in a way that is both meaningful and achievable? This paper makes three contributions:
\begin{enumerate}
    \item \textbf{Conceptual synthesis: } Dalenius's impossibility results are connected to modern fairness constraints, showing these are not independent problems but manifestations of fundamental information-theoretic limits.
    \item \textbf{Practical characterization: } Concrete bounds on the privacy-utility-fairness tradeoff are provided, including sample size requirement and group-specific error rates.
    \item \textbf{Decision framework: } Existing mitigation strategies are surveyed, and guidance for practitioners choosing between competing 
      approaches is provided.
\end{enumerate}

While individual results exist, no prior work provides this integrated treatment with worked examples spanning from theoretical foundations to deployment decisions.

\section{The Impossibility of Statistical Disclosure Control}
\label{sec:impossibility}

\subsection{Dalenius's Intuitive Goal}

In his now-famous 1977 article, Dalenius proposed an intuitive privacy guarantee \cite{dalenius_towards_1977}:

\begin{quote}
\emph{``Access to a statistical database should not enable one to learn anything about an individual that could not be learned without access.''}
\end{quote}

Formally, let $D$ be a database and $\mathcal{M}$ be a mechanism that releases information about $D$, such as the publication of survey results. For any individual $i$ and any property $P$, Dalenius's goal requires:
\begin{equation}
\Pr[P(i) | \mathcal{M}(D)] \approx \Pr[P(i)]
\end{equation}

To supplement this idea, he defined three modes of disclosure:
\begin{enumerate}
    \item \textbf{Identity Disclosure} occurs when a record in the database can be linked to a specific individual. For example, if an attacker can determine which database entry belongs to an individual in the database, by convention, we will refer to this individual as Alice.
    \item \textbf{Attribute Disclosure} is present when new, possibly sensitive, information about an individual can be determined. For example, if Alice's medical condition can be deduced from released statistics. 
    \item \textbf{Inferential Disclosure} occurs when specific characteristics can be intuited with higher accuracy than before. For example, if Alice is in a smoking related dataset which discloses that 99 of 100 individuals in her zip code smoke, one can guess with high accuracy that she is likely a smoker. 
\end{enumerate}

The concept of direct and indirect identifiers also further defines the statistical disclosure problem. Direct identifiers encompass any information that, on its own, can be linked directly to an individual. While it seems logical to simply remove this data in order to preserve privacy, the identification problem still exists. This is because indirect identifiers can also positively identify an individual when combined. For example, Perry found in her 2011 study that individuals can be positively identified with 87\% accuracy if their zip code, birth date, and gender are available \cite{perry_youre_2011}. This finding is concerning as this seemingly innocuous information, when combined, clearly violates Dalenius's initial privacy goal. 

While the idea of statistical disclosure control (SDC) is intuitively appealing, the definition collapses under formal analysis.

\subsection{Why Semantic Security Fails: A Concrete Example}

\begin{example}[The Smoking Dataset]
\label{ex:smoking}
Consider a medical database $D$ containing records of 1000 patients, including whether each smokes. Suppose:
\begin{itemize}
    \item 99\% of people with a rare genetic marker $G$ smoke
    \item Alice has genetic marker $G$ (publicly known via a genetic testing company)
    \item The database releases: ``52\% of patients smoke''
\end{itemize}

\textbf{Before seeing the database release:}
\begin{equation}
\Pr[\text{Alice smokes}] = 0.99 \quad \text{(from public genetic data)}
\end{equation}

\textbf{After seeing the database release:}
If Alice were \emph{not} in the database, the smoking rate would be significantly different from 52\% (since she likely smokes). This is because Alice’s presence (or absence) in the database produces a detectable, statistical shift as the population outside of the database is dominated by smokers with $G$. Therefore:
\begin{equation}
\Pr[\text{Alice is in database} \mid \text{release}] \approx 1
\end{equation}

The conditions of naive disclosure control have been met: only an aggregate statistic is released, not individual data. However, seeing the data makes it possible to learn that Alice is in a medical database, a piece of private information that wasn't known before seeing the statistical database. 
\end{example}

\subsection{The General Impossibility Result}

Example~\ref{ex:smoking} is not an edge case. Dwork proved that Dalenius's goal is fundamentally impossible \cite{dwork_differential_2006}.

Ultimately, it became clear that Dalenius's vision for statistical disclosure control faced three insurmountable challenges which arose from the very definition of SDC:
\begin{enumerate}
    \item Auxiliary information is ubiquitous; people know information about each other from sources outside the database. 
    \item "Learning nothing new" is ill-defined. The mechanism for measuring what an individual or entity "already knew" is unclear and restrictive. 
    \item Useful data analysis inherently reveals some information and if no information is revealed, the data isn't useful. This conflict introduces the first instance of the Privacy-Utility tradeoff. 
\end{enumerate}

The nature of these challenges can be combined to devise an informal theorem on the impossibility of SDC: 

\begin{theorem}[Informal]
Any mechanism that releases non-trivial information about a database violates SDC for some auxiliary information distribution.
\end{theorem}

This is because the database \emph{itself} is information; its existence changes what can be inferred. Statistical disclosure demands the impossible: that information is released without releasing information.

More formally, Dwork proved this as follows: 
\begin{theorem}[Impossibility of Semantic Security \cite{dwork_differential_2006}]
Let $\mathcal{M}$ be any mechanism that takes a database $D$ as input and produces output $M(D)$. Suppose $\mathcal{M}$ satisfies Dalenius's statistical disclosure control condition: for all individuals $i$, properties $P$, and auxiliary information $\mathcal{A}$,
\begin{equation}
\left| \Pr[P(i) \mid M(D), \mathcal{A}] - \Pr[P(i) \mid \mathcal{A}] \right| \leq \delta
\end{equation}
for some negligible $\delta > 0$.

Then either:
\begin{enumerate}
    \item $\mathcal{M}$ releases no information about $D$ (i.e., $M(D)$ is independent of $D$), or
    \item There exists auxiliary information $\mathcal{A}^*$ such that the statistical disclosure control condition is violated for some individual $i^*$ and property $P^*$.
\end{enumerate}
\end{theorem}

\begin{proof}[Proof Sketch]
Suppose $\mathcal{M}$ releases non-trivial information, meaning there exist databases $D_1, D_2$ such that the output distributions differ: $M(D_1) \not\equiv M(D_2)$. 

Let $P*(i)$ be the predicate representing whether individual $i$’s record appears in $D$. Take auxiliary information $\mathcal{A}^*$ to be complete knowledge of all individuals except $i$, plus the mechanism $\mathcal{M}$ itself.

Since $M(D_1)$ and $M(D_2)$ differ, observing $M(D)$ provides information about which database is being used. By Bayes' theorem:
\begin{equation}
\Pr[i \in D \mid M(D), \mathcal{A}^*] = \frac{\Pr[M(D) \mid i \in D, \mathcal{A}^*] \cdot \Pr[i \in D \mid \mathcal{A}^*]}{\Pr[M(D) \mid \mathcal{A}^*]}
\end{equation}

If the output distributions differ significantly when $i$ is present vs absent, the posterior probability $\Pr[i \in D \mid M(D), \mathcal{A}^*]$ will differ from the prior $\Pr[i \in D \mid \mathcal{A}^*]$, violating statistical disclosure control.
\end{proof}

\section{Differential Privacy: An Achievable Alternative}
\label{sec:dp}
In 2006, Dwork presented the idea of Differential Privacy (DP), a groundbreaking privacy metric that addresses Dalenius's impossibility by fundamentally shifting the goal of privacy. 

Rather than requiring that the database reveal \emph{nothing}, differential privacy ensures that any individual's presence or absence has \emph{minimal impact} on the output distribution.

Formally: 
\begin{definition}[Differential Privacy \cite{dwork_differential_2006}]
A randomized mechanism $\mathcal{M}: \mathcal{D} \to \mathcal{R}$ satisfies $\varepsilon$-differential privacy if for all neighboring databases $D, D'$ differing in one record, and all possible outputs $S \subseteq \mathcal{R}$:
\begin{equation}
\Pr[\mathcal{M}(D) \in S] \leq e^\varepsilon \cdot \Pr[\mathcal{M}(D') \in S]
\end{equation}
\end{definition}

\textbf{Interpretation:} Whether an individual is in the database or not changes the probability of any outcome by at most a multiplicative factor of $e^\varepsilon$.

\subsection{The Laplace Mechanism}

The Laplace mechanism is the canonical method for achieving differential privacy on counting queries.

\begin{definition}[Sensitivity]
The $\ell_1$-sensitivity of a function $f: \mathcal{D} \to \mathbb{R}^k$ is:
\begin{equation}
\Delta f = \max_{D, D' \text{ neighbors}} \|f(D) - f(D')\|_1
\end{equation}
\end{definition}

For counting queries, $\Delta f = 1$ since adding or removing one individual changes the count by at most 1.

\begin{algorithm}[h]
\caption{Laplace Mechanism}
\label{alg: laplace}
\begin{algorithmic}[1]
\REQUIRE Database $D$, query function $f$, privacy parameter $\varepsilon$
\STATE Compute sensitivity $\Delta f$
\STATE Compute true answer: $a = f(D)$
\STATE Sample noise: $\eta \sim \text{Lap}(\Delta f / \varepsilon)$
\RETURN $\tilde{a} = a + \eta$
\end{algorithmic}
\end{algorithm}

The Laplace distribution with scale parameter $b$ has probability density:
\begin{equation}
p(\eta \mid b) = \frac{1}{2b} e^{-|\eta|/b}
\end{equation}

\begin{theorem}[Privacy Guarantee of Laplace Mechanism \cite{dwork_differential_2009}]
Algorithm~\ref{alg: laplace} satisfies $\varepsilon$-differential privacy.
\end{theorem}

\begin{proof}
Let $D, D'$ be neighboring databases that differ in one record. Let $a = f(D)$ and $a' = f(D')$. By sensitivity, $|a - a'| \leq \Delta f$.

For any output set $S \subseteq \mathbb{R}$:
\begin{align}
\frac{\Pr[M(D) \in S]}{\Pr[M(D') \in S]} &= \frac{\int_S p(\eta - a \mid \Delta f/\varepsilon) d\eta}{\int_S p(\eta - a' \mid \Delta f/\varepsilon) d\eta} \\
&= \frac{\int_S e^{-\varepsilon|s - a|/\Delta f} ds}{\int_S e^{-\varepsilon|s - a'|/\Delta f} ds} \\
&\leq e^{\varepsilon |a - a'|/\Delta f} \\
&\leq e^{\varepsilon}
\end{align}
where the final inequality uses $|a - a'| \leq \Delta f$.
\end{proof}

\begin{figure}[h]
\centering
\begin{tikzpicture}[scale=1.0]
    % Draw axes
    \draw[->] (-4,0) -- (4,0) node[right] {$\eta$};
    \draw[->] (0,0) -- (0,3) node[above] {$p(\eta)$};
    
    % Draw Laplace distribution for epsilon = 1
    \draw[blue, thick, domain=-3.5:3.5, samples=100] 
        plot (\x, {0.5*exp(-abs(\x))}) node[right] {};
    \node[blue] at (2.5, 1.8) {$\varepsilon = 1$};
    
    % Draw Laplace distribution for epsilon = 0.5
    \draw[red, thick, domain=-3.5:3.5, samples=100] 
        plot (\x, {0.25*exp(-0.5*abs(\x))}) node[right] {};
    \node[red] at (2.5, 1.2) {$\varepsilon = 0.5$};
    
    % Draw Laplace distribution for epsilon = 2
    \draw[green!60!black, thick, domain=-3.5:3.5, samples=100] 
        plot (\x, {exp(-2*abs(\x))}) node[right] {};
    \node[green!60!black] at (2.5, 2.4) {$\varepsilon = 2$};
    
    % Mark key points
    \draw[dashed] (0,0) -- (0,2.5);
    \node[below] at (0,0) {$0$};
    
    % Add annotation
    \node[align=center] at (0, -0.8) {Smaller $\varepsilon$ = More privacy = More noise};
\end{tikzpicture}
\caption{Laplace noise distributions for different privacy parameters $\varepsilon$. The scale parameter is $b = 1/\varepsilon$ for sensitivity $\Delta f = 1$. Higher privacy (smaller $\varepsilon$) results in wider distributions and more noise added to the true answer.}
\label{fig:laplace}
\end{figure}
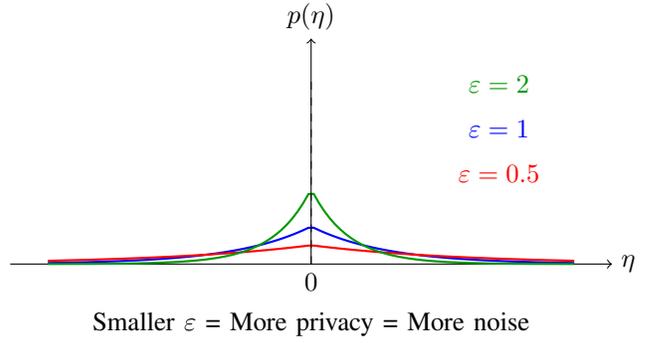

\subsection{Composition and Privacy Accounting}

Differential privacy is also advantageous because privacy guarantees \emph{compose}.

\begin{theorem}[Basic Composition \cite{dwork_differential_2006}]
If $\mathcal{M}_1$ satisfies $\varepsilon_1$-DP and $\mathcal{M}_2$ satisfies $\varepsilon_2$-DP, then running both on the same database satisfies $(\varepsilon_1 + \varepsilon_2)$-DP.
\end{theorem}

This enables \emph{privacy accounting}, or tracking cumulative privacy loss across multiple queries.

\subsection{Alternative Differential Privacy Mechanisms}

While the Laplace mechanism is foundational, several other mechanisms have been developed to optimize different objectives after Dwork's initial mechanism was released. Notably:

\begin{itemize}
    \item \textbf{Gaussian Mechanism} \cite{dwork_algorithmic_2014}: Dwork expands her own work by proposing the addition of Gaussian noise $\mathcal{N}(0, \sigma^2)$ where $\sigma = \frac{\sqrt{2\ln(1.25/\delta)} \cdot \Delta f}{\varepsilon}$ \footnote{Assuming the original $(\varepsilon,\delta)-DP$ bound is being used.}. This change provides $(\varepsilon, \delta)$-differential privacy, a relaxed variant that allows a small probability $\delta$ of privacy failure. The modification is now used in many modern deep learning applications like DP-SGD \cite{abadi_deep_2016}.
    
    \item \textbf{Exponential Mechanism} \cite{mcsherry_mechanism_2007}: For non-numeric queries (e.g., selecting the best element from a set), McSherry et al. propose sampling the output $r$ with the probability proportional to $\exp(\varepsilon \cdot u(D, r) / (2\Delta u))$, where $u$ is a utility function. This strategy is used for recommendations and ranking problems.
    
    \item \textbf{Sparse Vector Technique (SVT)} \cite{dwork_algorithmic_2014}: Dwork also proposes a mechanism to more efficiently answer many queries by only reporting when query results exceed a threshold. This is particularly useful for query optimization, where most queries return negative results.
    
    \item \textbf{Sample-and-Aggregate} \cite{nissim_smooth_2007}: Nissim et al. propose partitioning data into blocks, computing the function on each block, then aggregating the output with noise. This provides better accuracy when the underlying function is smooth or has low sensitivity on small datasets.
    
    \item \textbf{Propose-Test-Release (PTR)} \cite{dwork_differential_2009}: Dwork also proposed a system that checks whether data satisfies predetermined properties (e.g., enough samples) before releasing more accurate answers. This adapts the privacy-utility tradeoff to the actual dataset characteristics.
\end{itemize}

Each mechanism optimizes the privacy-utility tradeoff to address different scenarios, demonstrating that differential privacy is not a single strategy, but rather a framework that can be adapted to individual applications.

\section{The Privacy-Utility Tradeoff}
\label{sec:tradeoffs}

\subsection{Utility Under Statistical Uncertainty}

Differential privacy guarantees privacy but adds noise, reducing utility. Makhdoumi and Fawaz \cite{makhdoumi_privacy-utility_2013} characterized this tradeoff when the data distribution is unknown:

\textbf{Setup:} Let $\theta$ be an unknown parameter we want to estimate from database $D$.
\begin{itemize}
    \item \textbf{Utility:} Estimation error $\mathbb{E}[(\hat{\theta} - \theta)^2]$
    \item \textbf{Privacy:} Mutual information $I(\theta; \mathcal{M}(D))$ (alternative privacy metric)
\end{itemize}

\textbf{Key Result:} For any $\varepsilon$-DP mechanism:
\begin{equation}
\text{MSE}(\hat{\theta}) \geq \frac{\sigma^2_\theta}{n} + \frac{\Delta^2}{\varepsilon^2 n}
\end{equation}
where $\sigma^2_\theta$ is the inherent variance and $\Delta$ is sensitivity.

\textbf{Interpretation:} Privacy noise adds an unavoidable $O(1/\varepsilon^2)$ term to error. Better privacy (smaller $\varepsilon$) naturally results in diminished utility, as shown in Figure~\ref{fig: privacy-utility}.  

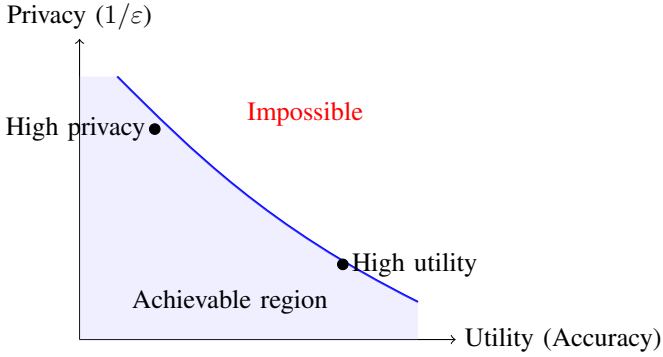
\begin{figure}[h]
\begin{center}
\begin{tikzpicture}[scale=1.0]
    % Axes
    \draw[->] (0,0) -- (5,0) node[right] {Utility (Accuracy)};
    \draw[->] (0,0) -- (0,4) node[above] {Privacy ($1/\varepsilon$)};
    
    % Pareto frontier curve
    \draw[thick, blue] (0.5,3.5) .. controls (1.5,2.5) and (2.5,1.5) .. (4.5,0.5);
    
    % Achievable region
    \fill[blue!20, opacity=0.3] (0.5,3.5) .. controls (1.5,2.5) and (2.5,1.5) .. (4.5,0.5) -- (4.5,0) -- (0,0) -- (0,3.5) -- cycle;
    
    % Impossible region
    \node at (3,3) {\textcolor{red}{Impossible}};
    
    % Sample points
    \filldraw[black] (1,2.8) circle (2pt) node[left] {High privacy};
    \filldraw[black] (3.5,1) circle (2pt) node[right] {High utility};
    
    % Labels
    \node at (2,0.5) {Achievable region};
\end{tikzpicture}
\caption{The Privacy-Utility Tradeoff}
\label{fig: privacy-utility}
\end{center}
\end{figure}

\section{Adding Fairness: The Three-Way Tradeoff}

\subsection{Fairness Metrics in ML}

Consider a binary classifier $h: \mathcal{X} \to \{0,1\}$ trained on private data. Let $A$ denote a sensitive attribute (e.g., race, gender).

\begin{definition}[Demographic Parity]
$h$ satisfies demographic parity if:
\begin{equation}
\Pr[\hat{Y} = 1 | A = 0] = \Pr[\hat{Y} = 1 | A = 1]
\end{equation}
\end{definition}

\begin{definition}[Equalized Odds]
$h$ satisfies equalized odds if for $y \in \{0,1\}$:
\begin{equation}
\Pr[\hat{Y} = 1 | Y = y, A = 0] = \Pr[\hat{Y} = 1 | Y = y, A = 1]
\end{equation}
\end{definition}

This leads to:
\begin{definition}[Fairness Violation Under Demographic Parity]
\begin{equation}
    F(\mathcal{M}) = |Pr[\hat{y}=1|A=0] - Pr[\hat{y}=1|A=1]|
\end{equation}
\end{definition}

\subsection{Privacy Degrades Fairness: Empirical Evidence}

Pannekoek and Spigler \cite{pannekoek_investigating_2021} trained 4 neural networks (simple NN, fair NN, DP NN, and DP-Fair NN) with DP-SGD and found. 
The privacy mechanism was implemented using DP-SGD (Differentially Private Stochastic Gradient Descent), where clip gradients were used to bound sensitivity. Calibrated Gaussian noise is then added to gradients and the privacy budget is tracked using a moments accountant. The fairness mechanisms made use of a demographic parity constraint where a fairness penalty term was added to the loss function. This minimized the difference in positive prediction rates across groups, resulting in a balance of fairness and accuracy. 

Their work, summarized in Table~\ref{tab:pannekoek_compact}, demonstrated that, as expected, noise injection for privacy protection degrades learning, especially for smaller privacy budgets (lower $\varepsilon$). Furthermore, the inclusion of DP mechanisms disproportionately impacts minority demographic groups, making the model less fair. There seemed to be an upper bound to this effect, however, as a certain degree of fairness constraints could be added without significantly degrading the model utility. Unfortunately, applying the DP-Fair model, which combined both privacy and fairness, showed the largest utility degradation.

\begin{table*}[h!]
\centering
\caption{Comprehensive comparison of models on Adult dataset ($\varepsilon=0.1$, $\delta=0.00001$ for DP models) reproduced from the work by Pannekoek and Spigler}
\label{tab:pannekoek_compact}
\begin{tabular}{llcccl}
\toprule
\textbf{Type} & \textbf{Model} & \textbf{Privacy} & \textbf{Accuracy (\%)} & \textbf{Risk Diff.} & \textbf{Notes} \\
\midrule
\multirow{4}{*}{Neural Net} 
& S-NN & None & $84.14 \pm 0.34$ & $0.131 \pm 0.015$ & Baseline \\
& F-NN & None & $79.25 \pm 3.50$ & $0.057 \pm 0.007$ & Fair (lenient) \\
& DP-NN & $\varepsilon=0.1$ & $84.03 \pm 0.05$ & $0.136 \pm 0.002$ & Private, unfair \\
& \textbf{DPF-NN} & $\varepsilon=0.1$ & $\mathbf{82.98 \pm 0.19}$ & $\mathbf{0.048 \pm 0.002}$ & \textbf{Private \& fair} \\
\midrule
\multirow{4}{*}{Log. Reg.} 
& LR & None & $83.80 \pm 0.23$ & $0.158 \pm 0.006$ & Baseline \\
& FairLR & None & $77.39 \pm 5.21$ & $0.010 \pm 0.007$ & Fairest \\
& PrivLR & $\varepsilon=0.1$ & $62.63 \pm 14.80$ & $0.088 \pm 0.081$ & Large variance \\
& PFLR* & $\varepsilon=0.1$ & $74.91 \pm 0.40$ & $0.003 \pm 0.004$ & Low accuracy \\
\bottomrule
\end{tabular}
\end{table*}

\subsection{Why Privacy and Fairness Conflict}

\textbf{Intuition:} Fairness metrics require an accurate estimation of group-specific statistics; meanwhile, privacy noise disproportionately affects smaller groups.

\begin{example}[Minority Group Amplification]
Suppose group $A=0$ has $n_0 = 10000$ samples and group $A=1$ has $n_1 = 100$ samples. To estimate $\Pr[\hat{Y}=1|A=a]$ with $\varepsilon$-DP:
\begin{align}
\text{Noise std for } A=0: &\quad \sigma_0 \propto 1/(\varepsilon \sqrt{n_0}) = 1/(100\varepsilon) \\
\text{Noise std for } A=1: &\quad \sigma_1 \propto 1/(\varepsilon \sqrt{n_1}) = 1/(10\varepsilon)
\end{align}

The minority group has 10× more relative noise, making fairness metrics less reliable.
\end{example}

\section{Toward Fundamental Limits: A Research Direction}

\subsection{The Central Question}
While experimental evidence exists, a unified theoretical
characterization remains incomplete. The gap in existing research presents the question: 
\emph{Can we characterize the \emph{Pareto frontier} of achievable $(\varepsilon, \text{utility}, \text{fairness})$ triples? }

Using this lemma justifies the existence of the Pareto frontier: 
\begin{observation}[Frontier Existence]
    The Pareto frontier exists and is non-empty for all $\varepsilon > 0$.
    
    By compactness, the multi-objective optimization problem has a solution. The Pareto frontier is the boundary of the achievable region.
\end{observation}

\subsection{Lower Bounds of Privacy Utility and Fairness}
The lower bound of the privacy-utility tradeoff for any $\varepsilon-DP$ mechanism $\mathcal{M}$ can be defined as \cite{makhdoumi_privacy-utility_2013}:
\begin{equation}
    U(\mathcal{M}) \geq U_0 - O(d/(\varepsilon n))
\end{equation}
where $U_0$ is utility without privacy, $d$ is dimensionality, $n$ is sample size

Furthermore, the privacy-fairness bound for minority groups with representation $p$ can be defined as \cite{cummings_compatibility_2019}:

\begin{equation}
    F(\mathcal{M}) \geq \Omega(1/(\varepsilon\sqrt{(np)}))
\end{equation}

Where DP noise amplifies for smaller groups as a natural result of a smaller sample size. 

This widely used lower bound demonstrates how private estimation of subgroup statistics requires noise of scale$1/(\varepsilon n)$ \cite{cummings_compatibility_2019}. The effective sample size for the subgroup, however, is $np$. Noise relative to subgroup magnitude scales as 

\begin{equation}
\mathrm{SE}_{\text{sampling}} \approx \sqrt{\frac{q(1-q)}{n_p}}, 
\qquad
\mathrm{SE}_{\text{DP-noise}} \approx \sqrt{2}\,\frac{1}{\varepsilon\, n_p}
\end{equation}
\begin{equation}
\frac{\mathrm{SE}_{\text{DP-noise}}}{\mathrm{SE}_{\text{sampling}}}
\approx
\frac{\sqrt{2}}{\varepsilon}\cdot\frac{1}{\sqrt{n_p}}\cdot\frac{1}{\sqrt{q(1-q)}}
\sim 
\frac{1}{\varepsilon\sqrt{n_p}}.
\end{equation}

\subsection{Impossibility Result}
The existence of the Pareto frontier and the lower bounds in equations (19) and (20) can be combined to create a critical sample size for the impossibility result:

For utility to remain above threshold $(U > 0.5)$ while maintaining 
fairness constraint $(F < F_{target})$, we require:

\begin{equation}
    d/(\varepsilon n) < U_0 - 0.5  \text{ AND }  1/(\varepsilon \sqrt{(np)}) < F_{target}
\end{equation}

Solving for n in both inequalities and taking the more restrictive 
bound yields:

\begin{equation}
    n^* = \Theta(d/(\varepsilon^2p·F^2))
\end{equation}

This threshold represents the sample size below which the privacy-
fairness-utility tradeoff becomes infeasible for the specified 
constraints. This shows that fairness, privacy, and utility jointly form a Pareto frontier: you cannot improve one without degrading at least one of the others unless you increase sample size.

This finding further emphasizes the need to either increase sample size or mitigate minority representation concerns through alternative methods, such as the addition of synthetic data. 

\subsection{The Significance of the Fairness Addition}
To emphasize the importance of the addition of fairness in differential privacy applications, consider the following case study: 

Consider a machine learning system designed to predict recidivism risk for criminal defendants, similar to the COMPAS system deployed in U.S. courts \cite{adrienne_brackey_analysis_2019}. The system uses historical criminal justice data containing sensitive information about individuals.

Let $D = \{(x_i, y_i, a_i)\}_{i=1}^n$ be a dataset where:
\begin{itemize}
    \item $x_i \in \mathbb{R}^d$ represents features (age, prior convictions, charge severity, etc.)
    \item $y_i \in \{0, 1\}$ indicates whether individual $i$ reoffended within 2 years
    \item $a_i \in \{0, 1\}$ is a protected attribute (e.g., race: White=0, Black=1)
\end{itemize}

The goal is to train a classifier $h: \mathbb{R}^d \to \{0,1\}$ that predicts recidivism risk while satisfying:
\begin{enumerate}
    \item \textbf{Privacy}: $\varepsilon$-differential privacy to protect individual records
    \item \textbf{Utility}: High accuracy, measured by $\text{Acc} = \Pr[\hat{Y} = Y]$
    \item \textbf{Fairness}: Equalized false positive rates across racial groups
\end{enumerate}

This introduces a key conflict. Suppose the dataset has:
\begin{itemize}
    \item Group $A=0$ (White): $n_0 = 5000$ samples, base rate $p_0 = 0.30$
    \item Group $A=1$ (Black): $n_1 = 500$ samples, base rate $p_1 = 0.45$
\end{itemize}

\textbf{Without Privacy:} A standard logistic regression achieves:
\begin{align}
\text{Acc}_0 &= 0.82, \quad \text{FPR}_0 = 0.12 \\
\text{Acc}_1 &= 0.79, \quad \text{FPR}_1 = 0.15
\end{align}

Fairness violation: $|\text{FPR}_1 - \text{FPR}_0| = 0.03$

\textbf{With $\varepsilon=1$ Differential Privacy:} Using DP-SGD \cite{abadi_deep_2016} for training:
\begin{align}
\text{Acc}_0 &= 0.77, \quad \text{FPR}_0 = 0.18 \\
\text{Acc}_1 &= 0.68, \quad \text{FPR}_1 = 0.28
\end{align}

Fairness violation increased: $|\text{FPR}_1 - \text{FPR}_0| = 0.10$

The minority group experiences:
\begin{itemize}
    \item 11\% accuracy drop (0.79 $\to$ 0.68) vs 5\% for majority (0.82 $\to$ 0.77)
    \item 87\% increase in false positive rate (0.15 $\to$ 0.28) vs 50\% for majority
\end{itemize}

The standard error in estimating group-specific statistics scales as:
\begin{equation}
\text{SE}(a) \propto \frac{1}{\sqrt{n_a \cdot \varepsilon^2}}
\end{equation}

In the above example:
\begin{align}
\text{SE}(A=0) &\propto \frac{1}{\sqrt{5000 \cdot 1}} \approx 0.014 \\
\text{SE}(A=1) &\propto \frac{1}{\sqrt{500 \cdot 1}} \approx 0.045
\end{align}

The minority group has $3.2\times$ higher estimation error, directly translating to worse fairness metrics.

This illustrative case study demonstrates that:
\begin{enumerate}
    \item Privacy mechanisms amplify existing disparities in data representation
    \item Fair outcomes require either: (a) larger minority samples, (b) weaker privacy, or (c) accepting lower overall utility
    \item In high-stakes domains like criminal justice, these tradeoffs have serious real-world consequences
\end{enumerate}

The fundamental issue is not a flaw in differential privacy, but rather an inherent mathematical constraint due to the nature of conflicting definitions: protecting privacy through noise injection necessarily degrades the statistical reliability of estimates, with disproportionate impact on groups with less data representation.

\section{Practical Strategies for Navigating Tradeoffs}

While the theoretical impossibility results demonstrate fundamental limits, practitioners have developed several strategies to navigate the privacy-utility-fairness tradeoff space in real applications.

\subsection{Data Augmentation and Synthetic Data}

One common approach to mitigate minority group disadvantage is to augment the training data itself:

\begin{itemize}
    \item \textbf{Differentially Private Synthetic Data Generation} \cite{lin_differentially_2025}: Synthetic records that preserve statistical properties can be generated while satisfying DP. Tools like PrivBayes or DP-GAN can create artificial minority group samples to balance datasets before training.
    
    \item \textbf{Oversampling with Privacy Budget Allocation}: Allocating a larger privacy budget ($\varepsilon$) to minority groups, can also address fairness. However, this requires careful consideration of group privacy guarantees \cite{kearns_preventing_2018}.
\end{itemize}

While synthetic data seems like a promising solution, it may not capture true distributional properties, potentially reducing model utility on real data \cite{watson_comprehensive_2024}.

\subsection{Adaptive Privacy Budgets}

Rather than uniform $\varepsilon$ across all groups:

\begin{itemize}
    \item \textbf{Group-Specific Privacy Parameters}: Allow different groups to have different privacy levels based on their representation and vulnerability \cite{cummings_compatibility_2019}. For example, set $\varepsilon_{\text{minority}} = 2$ and $\varepsilon_{\text{majority}} = 0.5$.
    
    \item \textbf{Privacy Amplification by Subsampling}: In mini-batch training, subsampling amplifies privacy, providing stronger guarantees for overrepresented groups without additional cost \cite{balle_privacy_2018}.
\end{itemize}

Unfortunately, differential treatment of groups also raises ethical questions about who deserves more privacy protection.

\subsection{Post-Processing for Fairness}

After training a DP model, fairness corrections can be applied to circumvent addressing the tradeoff in the model-building stage:

\begin{itemize}
    \item \textbf{Threshold Optimization} \cite{hardt_equality_2016}: Adjust decision thresholds per group to satisfy fairness constraints. Since post-processing preserves DP, this doesn't degrade privacy.
    
    \item \textbf{Calibration}: Ensure predicted probabilities are well-calibrated across groups, improving both fairness and reliability \cite{pleiss_fairness_2017}.
\end{itemize}

Trade-off: Post-processing can reduce overall accuracy and may not address underlying model biases.

\subsection{Architecture and Algorithm Design}

Design models that are inherently more robust to privacy noise:

\begin{itemize}
    \item \textbf{Fair Representations} \cite{zemel_learning_2013}: Learn intermediate representations that are statistically independent of sensitive attributes before applying DP mechanisms.
    
    \item \textbf{Adaptive Clipping in DP-SGD} \cite{andrew_differentially_2022}: Dynamically adjust gradient clipping thresholds to reduce noise while maintaining privacy, improving utility especially for minority groups.
    
    \item \textbf{Transfer Learning}: Pre-train on public data, then fine-tune with DP on sensitive data, reducing the privacy budget needed for good performance \cite{papernot_semi-supervised_2017}.
\end{itemize}

Trade-off: More complex architectures may be harder to interpret and validate.

\subsection{Choosing the Right Metric}

Different contexts call for different priorities:

\begin{itemize}
    \item \textbf{High-stakes decisions} (criminal justice, lending): Prioritize fairness and interpretability, accept lower privacy ($\varepsilon > 1$) or increased data collection requirements.
    
    \item \textbf{Medical research}: Prioritize strong privacy ($\varepsilon < 1$), accept reduced utility, use federated learning to increase effective sample size \cite{mcmahan_communication-efficient_2023}.
    
    \item \textbf{Personalized services}: Balance all three with moderate privacy ($\varepsilon \approx 1-3$), group-specific adjustments, and transparency about limitations.
\end{itemize}

Figure~\ref{fig: strategy-selection} demonstrates a decision tree for selecting the optimal strategy for common applications. 

\begin{figure*}[h!]
    \centering
    \includegraphics[width=1.0\linewidth]{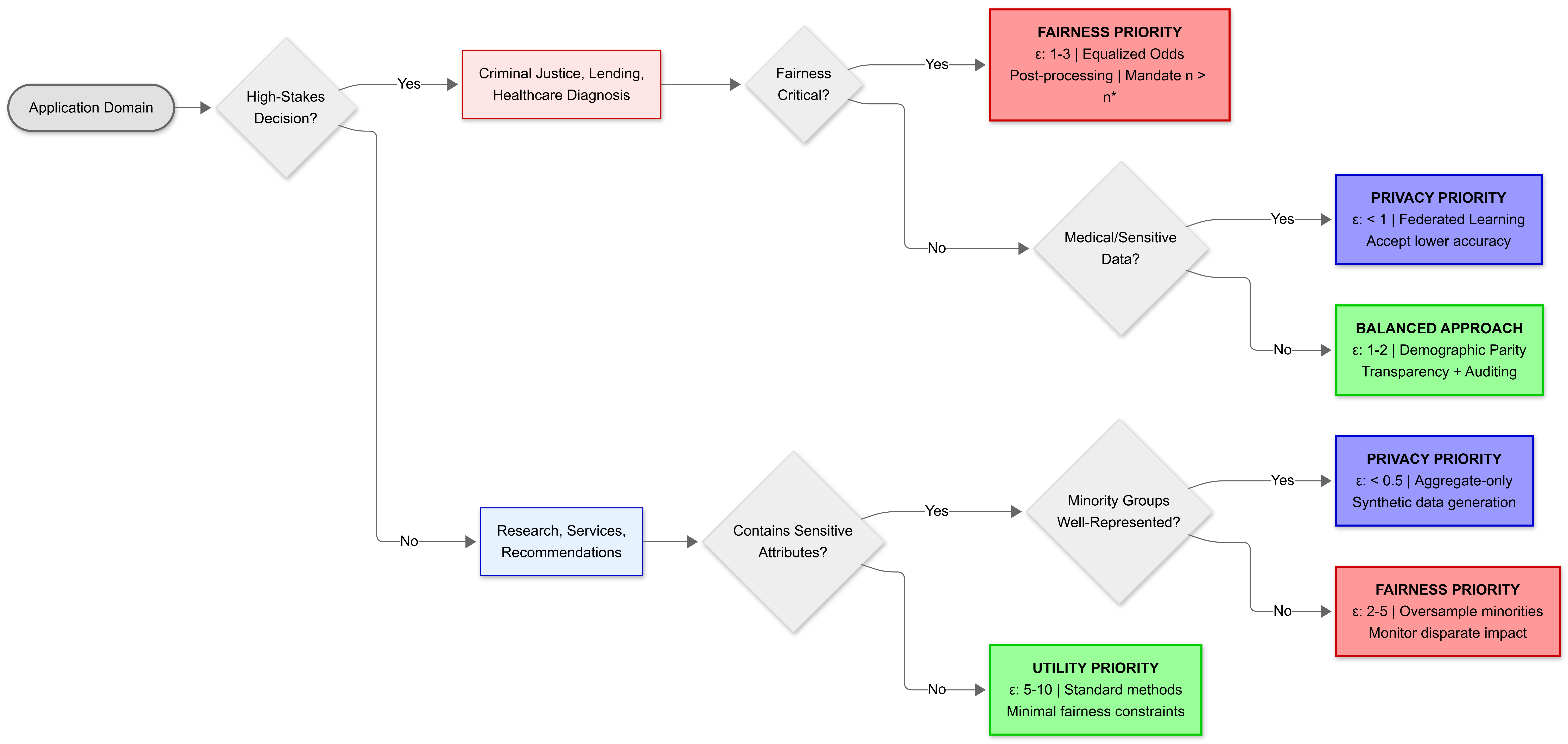}
    \caption{A Decision Tree for Tradeoff Mitigation Selection}
    \label{fig: strategy-selection}
\end{figure*}

\section{Conclusion}

The evolution from Dalenius's theory of Statistical Disclosure Control to modern differential privacy reveals a central pattern in privacy research: intuitive goals often prove impossible. Ultimately, some principles must be relaxed to create usable metrics. Differential privacy has become widely used not by achieving perfect privacy, as early privacy researchers thought would be necessary, but by providing a mathematically rigorous, achievable guarantee.

This success, of course, comes with costs: privacy trades off against utility and fairness in fundamental ways. Understanding these limits and characterizing what is and isn't achievable remains an active area of research with both theoretical and practical implications.

As machine learning systems increasingly operate on sensitive data, understanding these fundamental limits becomes critical for both algorithm design and policy decisions.

% References
\bibliography{references.bib}

\end{document}